\documentclass[11pt]{article}%
\usepackage{amsfonts}
\usepackage{amsmath}
\usepackage{amssymb}
\usepackage{graphicx}%
\providecommand{\U}[1]{\protect\rule{.1in}{.1in}}
\newtheorem{theorem}{Theorem}
\newtheorem{acknowledgement}[theorem]{Acknowledgement}

\newtheorem{definition}[theorem]{Definition}

\newtheorem{lemma}[theorem]{Lemma}

\newtheorem{proposition}[theorem]{Proposition}
\newtheorem{remark}[theorem]{Remark}

\newenvironment{proof}[1][Proof]{\noindent\textbf{#1.} }{\ \rule{0.5em}{0.5em}}
\begin{document}

\title{Toeplitz Density Operators and their Separability Properties}
\author{Maurice de Gosson\thanks{maurice.de.gosson@univie.ac.at}\\University of Vienna\\Faculty of Mathematics (NuHAG)}
\maketitle

\begin{abstract}
Toeplitz operators (also called localization operators) are a generalization
of the well-known anti-Wick pseudodifferential operators studied by Berezin
and Shubin. When a Toeplitz operator is positive semi-definite and has trace
one we call it a density Toeplitz operator. Such operators represent physical
states in quantum mechanics. In the present paper we study several aspects of
Toeplitz operators when their symbols belong to some well-known functional
spaces (e.g. the Feichtinger algebra) and discuss (tentatively) their
separability properties.

\end{abstract}



\section{Introduction}

There is a vast mathematical literature on Toeplitz operators and their
variants (generalized anti-Wick operators), but these operators are much less
known and used in quantum mechanics. This is a kind of paradox since Toeplitz
operators were advertised and developed under the influence of Berezin and
Shubin \cite{Berezin,beshu} in the context of quantization. Certain particular
cases are however known to most quantum physicists under the name of
\textquotedblleft anti-Wick quantization\textquotedblright\ or
\textquotedblleft Husimi function\textquotedblright. Still, the theory of
Toeplitz operators is much better known and more often used in the related
discipline of time-frequency analysis; among many references the reader might
want to consult the following papers
\cite{Balasz,boco1,bocogr,cogr,cogr0,cogr1,coro,grotoft,GH04,Luef,pili,Toft,togr}
to get an idea of what is going on in the field (I refrain from opposing
quantum mechanics and time-frequency analysis, one supposedly being a physical
theory and the other a mathematical theory, since both concern themselves wit
physical objects. It is just their aims and \textquotedblleft
philosophy\textquotedblright\ which differ). The aim of this paper is not to
review quite generally the theory of Toeplitz operators, but more modestly to
focus on the case where such operators are of trace class, more precisely
density operators: a density operator on a complex separable Hilbert space
$\mathcal{H}$ is a positive semidefinite trace class operator $\widehat{\rho}$
with trace $\operatorname*{Tr}(\widehat{\rho})=1$. We assume in this paper
that $\mathcal{H}=L^{2}(\mathbb{R}^{n})$. The following properties of density
operators are well-known: \textit{(i)} $\widehat{\rho}$ is self-adjoint;
\textit{(ii)} $\widehat{\rho}$ is the product of two Hilbert--Schmidt
operators (and hence compact); \textit{(iii)} $\widehat{\rho}$ is positive
semidefinite: $\widehat{\rho}$ $\geq0$. By the spectral theorem, there exists
an orthonormal basis $(\phi_{j})_{j}$ of $L^{2}(\mathbb{R}^{n})$ and
coefficients satisfying $\lambda_{j}\geq0$ and $\sum_{j}\lambda_{j}=1$ such
that $\widehat{\rho}$ can be written as a convex sum $\sum_{j}\lambda
_{j}\widehat{\Pi}_{\phi_{j}}$ of orthogonal projections $\widehat{\Pi}%
_{\phi_{j}}:$ $L^{2}(\mathbb{R}^{n})\longrightarrow\mathbb{C}\phi_{j}$
converging in the strong operator topology. The importance of density
operators in quantum mechanics comes from the fact that they represent (and
are identified with) \textquotedblleft mixed quantum states\textquotedblright;
these are mixtures of $L^{2}$-normalized \textquotedblleft pure
states\textquotedblright\ $(\psi_{j})_{j}$ in $L^{2}(\mathbb{R}^{n})$ weighted
by probabilities $\mu_{j}\geq0$ summing up to one; the corresponding mixed
state is then by definition the operator $\widehat{\rho}=\sum_{j}\mu
_{j}\widehat{\Pi}_{\psi_{j}}$ and represents the maximal knowledge one has
about the system under consideration. It is not difficult \cite{Birkbis,QHA}
to check that the operator $\widehat{\rho}$ thus defined indeed is a density
operator; note that the decomposition $\sum_{j}\mu_{j}\widehat{\Pi}_{\psi_{j}%
}$ of $\widehat{\rho}$ has no reason to be unique (Jayne's theorem, see
however \cite{CMdG1} where we compare different expansions of pure states). A
density operator is \textit{de facto} a Weyl operator in view of Schwartz's
kernel theorem; its Weyl symbol is $(2\pi\hbar)^{n}\rho$ where $\rho$ is the
\textquotedblleft Wigner distribution of $\widehat{\rho}$\textquotedblright%
\ defined by
\begin{equation}
\rho=\sum_{j}\mu_{j}W\psi_{j}\label{MGro1}%
\end{equation}
the series being convergent in the $L^{2}$-norm. Here $W\psi_{j}$ is the usual
Wigner transform of $\psi_{j}$. Consider now, as we did in \cite{ATFA}, a
family of functions
\[
\phi_{z_{\lambda}}(x,p)=e^{\frac{i}{\hbar}p_{\lambda}(x-x_{\lambda})}%
\phi(x-x_{\lambda})
\]
the $z_{\lambda}=(x_{\lambda},p_{\lambda})$ belonging to some lattice
$\Lambda\subset\mathbb{R}^{2n}$ and $\phi\in L^{2}(\mathbb{R}^{n})$ (hereafter
called \textquotedblleft window\textquotedblright) is a fixed function with
unit $L^{2}$-norm. We can define a corresponding density operator by
\begin{equation}
\widehat{\rho}=\sum_{z_{\lambda}\in\Lambda}\mu_{z_{\lambda}}\widehat{\Pi
}_{\phi_{z_{\lambda}}}\label{rhofila}%
\end{equation}
and its Wigner distribution is given by
\begin{equation}
\rho(z)=\sum_{z_{\lambda}\in\Lambda}\mu_{z_{\lambda}}W\phi(z-z_{\lambda})
\end{equation}
in view of the translational properties of the Wigner transform. We can
rewrite the formula above as the convolution product
\begin{equation}
\rho=\mu\ast W\phi\text{ \ , \ }\mu=\sum_{z_{\lambda}\in\Lambda}%
\mu_{z_{\lambda}}\delta_{z_{\lambda}}\label{discrete1}%
\end{equation}
where $\delta_{z_{\lambda}}$ is the Dirac measure on $\mathbb{R}^{2n}$
centered at $z_{\lambda}$. This suggests to consider more general
operators\ with Weyl symbol $(2\pi\hbar)^{n}\rho$ with $\rho=\mu\ast W\phi$
where $\mu$ is a Borel measure on $\mathbb{R}^{2n}$. The aim of this paper is
to study the properties of such operators; when they are density operators we
will call them \emph{Toeplitz density operators}. Such operators were first
considered by Berezin \cite{Berezin,beshu}, and have been developed since then
by various authors, sometimes under the name of \textit{localization
operators}, in time-frequency analysis
\cite{boco1,bocogr,coro,grotoft,Luef,pili}

\paragraph*{Notation}

We write $x=(x_{1},...,x_{n})$ and $p=(p_{1},...,p_{n})$; we will use the
notation $px$ for the inner product $p_{1}x_{1}+\cdot\cdot\cdot+p_{n}x_{n}$.
\textit{The scalar product on} $L^{2}(\mathbb{R}^{n})$ \textit{is defined by}%
\begin{equation}
(\psi|\phi)_{L^{2}}=\int_{\mathbb{R}^{n}}\psi(x)\overline{\phi(x)}dx.
\label{L2}%
\end{equation}
\textit{The space} $T^{\ast}\mathbb{R}^{n}\equiv\mathbb{R}^{2n}$ \textit{will
be equipped with the canonical symplectic structure} $\sigma=\sum_{j=1}%
^{n}dp_{j}\wedge dx_{j}$, \textit{given in matrix notation by} $\sigma
(z,z^{\prime})=Jz\cdot z^{\prime}$ \textit{where} $J=%
\begin{pmatrix}
0 & I\\
-I & 0
\end{pmatrix}
$ \textit{is the standard symplectic matrix.}

\section{Weyl and Anti-Wick Operators}

We are using here the notation in \cite{Birkbis} which the Reader can consult
for details and proofs.

\subsection{Weyl pseudodifferential operators\label{subsec11}}

Recall that the cross-Wigner transform of $(\psi,\phi)\in L^{2}(\mathbb{R}%
^{n})\times L^{2}(\mathbb{R}^{n})$ is defined by the absolutely convergent
integral
\begin{equation}
W(\psi,\phi)(x,p)=\left(  \tfrac{1}{2\pi\hbar}\right)  ^{n}\int_{\mathbb{R}%
^{n}}e^{-\tfrac{i}{\hbar}py}\psi(x+\tfrac{1}{2}y)\overline{\phi(x-\tfrac{1}%
{2}y)}dy~; \label{CrossW}%
\end{equation}
in particular $W(\psi,\psi)=W\psi$ is the usual Wigner transform. We have the
important relation:(\cite{Birkbis}, \S 9.2)
\begin{equation}
\int_{\mathbb{R}^{2n}}W(\psi,\phi)(z)dz=(\psi|\phi)_{L^{2}}. \label{marge}%
\end{equation}
Let $a\in\mathcal{S}^{\prime}(\mathbb{R}^{2n})$; the Weyl operator
$\widehat{A}=\operatorname*{Op}\nolimits_{\mathrm{W}}(a)$ with symbol $a$ is
(by definition) the unique operator $\widehat{A}:\mathcal{S}(\mathbb{R}%
^{n})\longrightarrow\mathcal{S}^{\prime}(\mathbb{R}^{n})$ such that%
\begin{equation}
\langle\widehat{A}\psi,\overline{\phi}\rangle=\langle\langle a,W(\psi
,\phi)\rangle\rangle\label{defA}%
\end{equation}
for all $(\psi,\phi)\in\mathcal{S}(\mathbb{R}^{n})\times\mathcal{S}%
(\mathbb{R}^{n})$ where $\langle\cdot,\cdot\rangle$ (\textit{resp}.
$\langle\langle\cdot,\cdot\rangle\rangle$) is the distributional bracket on
$\mathbb{R}^{n}$ (\textit{resp}. $\mathbb{R}^{2n}$). Using the Heisenberg
displacement operator%
\[
\widehat{T}(z_{0})\psi(x)=e^{\frac{i}{\hbar}(p_{0}x-\frac{1}{2}p_{0}x_{0}%
)}\psi(x-x_{0})
\]
we have the harmonic decomposition%
\begin{equation}
\widehat{A}=\left(  \tfrac{1}{2\pi\hbar}\right)  ^{n}\int_{\mathbb{R}^{2n}%
}a_{\sigma}(z)\widehat{T}(z)dz \label{weyl1}%
\end{equation}
where $a_{\sigma}$ is the symplectic Fourier transform of $a$:%
\begin{equation}
a_{\sigma}(z)=\left(  \tfrac{1}{2\pi\hbar}\right)  ^{n}\int_{\mathbb{R}^{2n}%
}e^{-\frac{i}{\hbar}\sigma(z,z^{\prime})}a(z)dz. \label{sft}%
\end{equation}
When $a\in\mathcal{S}(\mathbb{R}^{2n})$ we get the familiar textbook
definition%
\begin{equation}
\widehat{A}\psi(x)=\left(  \tfrac{1}{2\pi\hbar}\right)  ^{n}\int%
_{\mathbb{R}^{n}\times\mathbb{R}^{n}}e^{-\tfrac{i}{\hbar}p(x-y)}a(\tfrac{1}%
{2}(x+y),p)\psi(y)dpdy \label{weyl2}%
\end{equation}
valid for $\psi\in\mathcal{S}(\mathbb{R}^{n})$. In particular, the
distributional kernel of $\widehat{A}$ is
\begin{equation}
K(x,y)=\left(  \tfrac{1}{2\pi\hbar}\right)  ^{n}\int_{\mathbb{R}^{n}%
}e^{-\tfrac{i}{\hbar}p(x-y)}a(\tfrac{1}{2}(x+y),p)dp \label{kxy}%
\end{equation}
hence, using the Fourier inversion formula,%
\begin{equation}
a(x.p)=\int_{\mathbb{R}^{n}}e^{-\tfrac{i}{\hbar}py}K(x+\tfrac{1}{2}%
y,x-\tfrac{1}{2}y)dy. \label{axp}%
\end{equation}
We will also need the notion of transpose of a Weyl operator. Let
$\widehat{A}:\mathcal{S}(\mathbb{R}^{n})\longrightarrow\mathcal{S}%
(\mathbb{R}^{n})$ be a linear operator; the transpose $\widehat{A}^{T}$ of
$\widehat{A}$ is the unique operator $\mathcal{S}(\mathbb{R}^{n}%
)\longrightarrow\mathcal{S}(\mathbb{R}^{n})$ such that $\langle\widehat{A}%
\psi,\phi\rangle=\langle\psi,\widehat{A}^{T}\phi\rangle$ for all $(\psi
,\phi)\in\mathcal{S}(\mathbb{R}^{n})\times\mathcal{S}(\mathbb{R}^{n})$. If
$\widehat{A}$ has kernel $(x,y)\longmapsto K(x,y)$ then $\widehat{A}^{T}$ has
the kernel $(x,y)\longmapsto K(y,x)$ hence, in view of (\ref{axp}),%
\begin{equation}
\widehat{A}^{T}=\operatorname*{Op}\nolimits_{\mathrm{W}}(a\circ\overline
{I})\text{ \ , \ }\overline{I}(x,p)=(x,-p). \label{transp}%
\end{equation}

Weyl pseudo-differential operators enjoy the property of symplectic
covariance: let $\operatorname*{Sp}(n)$ be the standard symplectic group of
$\mathbb{R}^{2n}$. It is the group of all linear automorphisms of
$\mathbb{R}^{2n}$ such that $S^{\ast}\sigma=\sigma$ where $\sigma(z,z^{\prime
})=p\cdot x^{\prime}-p^{\prime}\cdot x$ if $z=(x,p)$, $z^{\prime}=(x^{\prime
},p^{\prime})$. We denote by $\operatorname*{Mp}(n)$ the unitary
representation in $L^{2}(\mathbb{R}^{n})$ of the double cover of
$\operatorname*{Sp}(n)$; $\operatorname*{Mp}(n)$ is the metaplectic group
\cite{Birkbis}; every $S\in\operatorname*{Sp}(n)$ is thus the projection of
two elements $\pm\widehat{S}$ of $\operatorname*{Mp}(n)$. We have
(\cite{Birkbis}, \S 10.3.1):%
\begin{equation}
\operatorname*{Op}\nolimits_{\mathrm{W}}(a\circ S^{-1})=\widehat{S}%
\operatorname*{Op}\nolimits_{\mathrm{W}}(a)\widehat{S}^{-1}. \label{sympco}%
\end{equation}

Here is a general result for the calculation of the trace of a Weyl operator:

\begin{proposition}
\label{propduwong}Let $\widehat{A}=\operatorname*{Op}\nolimits_{\mathrm{W}%
}(a)$ be a trace class operator. If $a\in L^{1}(\mathbb{R}^{2n})$ then%
\begin{equation}
\operatorname*{Tr}(\widehat{A})=\left(  \tfrac{1}{2\pi\hbar}\right)  ^{n}%
\int_{\mathbb{R}^{2n}}a(z)dz. \label{duwong1}%
\end{equation}

\end{proposition}

We will give refinements of this statement in Propositions \ref{PropTrace} and
\ref{Prop4} using the so-called Feichtinger algebra.

Notice that Proposition \ref{propduwong} requires that we know from the
beginning that $\widehat{A}$ is of trace class. We can get stronger statement
if we assume that the symbol $a$ belongs to some appropriate Shubin class
\cite{sh87}. A function $a\in C^{\infty}(\mathbb{R}^{2n})$ belongs to the
Shubin class $\Gamma_{\delta}^{m}(\mathbb{R}^{2n})$ if for every $\alpha
\in\mathbb{N}^{2n}$ there exists a constant $C_{\alpha}\geq0$ such that
\begin{equation}
|\partial_{z}^{\alpha}a(z)|\leq C_{\alpha}(1+|z|)^{m-\delta|\alpha|}\text{ for
}z\in\mathbb{R}^{2n}. \label{est1}%
\end{equation}

\begin{proposition}
\label{PropTrace}Let $\widehat{A}=\operatorname*{Op}\nolimits_{\mathrm{W}}(a)$
with $a\in\Gamma_{\delta}^{m}(\mathbb{R}^{2n})$. If $m<-2n$ then $\widehat{A}$
is of trace class and we have
\begin{equation}
\operatorname*{Tr}(\widehat{A})=\left(  \tfrac{1}{2\pi\hbar}\right)  ^{n}%
\int_{\mathbb{R}^{2n}}a(z)dz=a_{\sigma}(0) \label{trashub}%
\end{equation}
where $a_{\sigma}=F_{\sigma}a$ is the symplectic Fourier transform (\ref{sft})
of the symbol $a$.
\end{proposition}

See \cite{Birkbis}, \S 1.2.3, for a discussion of various trace formulas
occurring in the literature.

Gaussian functions of the type%
\begin{equation}
\rho(z)=\left(  \tfrac{1}{2\pi}\right)  ^{n}(\det\Sigma)^{-1/2}e^{-\frac{1}%
{2}\Sigma^{-1}z\cdot z} \label{Gaussian}%
\end{equation}
where $\Sigma$ is a real positive definite $2n\times2n$ matrix clearly satisfy
the conditions in Proposition \ref{PropTrace} and we have
\begin{equation}
\int_{\mathbb{R}^{2n}}\rho(z)dz=1. \label{norm}%
\end{equation}
However, the operator $\widehat{\rho}=(2\pi\hbar)^{n}\operatorname*{Op}%
\nolimits_{\mathrm{W}}(\rho)$ does not qualify as a density operator unless
the matrix $\Sigma$ satisfies the condition%
\begin{equation}
\Sigma+\frac{i\hbar}{2}J\geq0 \label{narcow1}%
\end{equation}
where \textquotedblleft\ $\geq0$\textquotedblright\ stands for
\textquotedblleft positive semi-definite\textquotedblright; this condition
ensures the positivity of $\widehat{\rho}$ \cite{Birkbis,Narcowich}. It can be
shown \cite{Birkbis} that condition (\ref{narcow1}) is a symplectically
invariant reformulation of the uncertainty principle of quantum mechanics
\cite{Birkbis,QHA}.

\subsection{Anti-Wick operators\label{subsec12}}

There are several ways to define anti-Wick operators; in \cite{sh87} Shubin
uses the following definition: given a symbol $a\in\mathcal{S}(\mathbb{R}%
^{n})$ the associated anti-Wick operator $\widehat{A}_{\mathrm{AW}%
}=\operatorname*{Op}\nolimits_{\mathrm{AW}}(a)$ is, by definition,
\begin{equation}
\widehat{A}_{\mathrm{AW}}=\int_{\mathbb{R}^{2n}}a(z)\widehat{\Pi}_{0}(z)dz
\label{awick1}%
\end{equation}
where $\widehat{\Pi}_{0}(z):L^{2}(\mathbb{R}^{n})\longrightarrow
L^{2}(\mathbb{R}^{n})$ is the orthogonal projection onto the ray generated by
$\widehat{T}(z)\phi_{0}$ where $\phi_{0}$ the standard Gaussian%
\begin{equation}
\phi_{0}(x)=(\pi\hbar)^{-n/4}e^{-|x|^{2}/2\hbar}. \label{standard}%
\end{equation}
This action of this projection is explicitly given by%
\begin{equation}
\widehat{\Pi}_{0}(z)\psi=(\psi|\widehat{T}(z)\phi_{0})_{L^{2}}\widehat{T}%
(z)\phi_{0} \label{piz}%
\end{equation}
and hence the operator $\widehat{A}_{\mathrm{AW}}$ is given by%
\begin{equation}
\widehat{A}_{\mathrm{AW}}\psi=\int_{\mathbb{R}^{2n}}a(z)(\psi|\widehat{T}%
(z)\phi_{0})_{L^{2}}\widehat{T}(z)\phi_{0}. \label{awick0}%
\end{equation}
We observe that (\cite{Birkbis}, \S 11.4.1) $(\psi|\widehat{T}(z)\phi
_{0})_{L^{2}}$ is, up to a factor, the radar ambiguity transform \cite{Gro} of
the pair $(\psi,\phi_{0})$; in fact
\begin{equation}
(\psi|\widehat{T}(z)\phi_{0})_{L^{2}}=(2\pi\hbar)^{n}\operatorname{Amb}%
(\psi,\phi_{0}) \label{psid}%
\end{equation}
where, by definition,
\begin{equation}
\operatorname{Amb}(\psi,\phi_{0})(z)=\left(  \tfrac{1}{2\pi\hbar}\right)
^{n}\int_{\mathbb{R}^{n}}{}^{-\tfrac{i}{\hbar}py}\psi(y+\tfrac{1}%
{2}x)\overline{\phi_{0}(y-\tfrac{1}{2}x)}dy. \label{crossamb}%
\end{equation}
Formula (\ref{awick0}) can thus be rewritten%
\begin{equation}
\widehat{A}_{\mathrm{AW}}\psi=(2\pi\hbar)^{n}\int_{\mathbb{R}^{2n}%
}a(z)\operatorname{Amb}(\psi,\phi_{0})(z)\widehat{T}(z)\phi_{0}.
\label{awick2}%
\end{equation}
Recall (\cite{Birkbis}, \S 9.3) the following simple relation between
$\operatorname{Amb}(\psi,\phi)$ and the cross-Wigner transform:%
\begin{equation}
\operatorname{Amb}(\psi,\phi)(z)=2^{-n}W(\psi,\phi^{\vee})(\tfrac{1}{2}z)
\label{ambwig}%
\end{equation}
where $\phi^{\vee}(x)=\phi(-x)$. It follows that (\ref{psid}) can be
rewritten, since $\phi_{0}^{\vee}=\phi_{0}$,
\begin{equation}
(\psi|\widehat{T}(z)\phi_{0})_{L^{2}}=(\pi\hbar)^{n}W(\psi,\phi_{0})(\tfrac
{1}{2}z)~ \label{psidbis}%
\end{equation}
so that we also have%
\begin{equation}
\widehat{A}_{\mathrm{AW}}\psi=(\pi\hbar)^{n}\int_{\mathbb{R}^{2n}}%
a(z)W(\psi,\phi_{0})(\tfrac{1}{2}z)\widehat{T}(z)\phi_{0}. \label{awick3}%
\end{equation}

The following characterization in terms of Weyl operators is often taken as a
definition of anti-Wick quantization:

\begin{proposition}
\label{Prop2}Let $a\in L^{1}(\mathbb{R}^{n})$. The Weyl symbol $b$ of the
operator $\widehat{A}=\operatorname*{Op}\nolimits_{\mathrm{AW}}(a)$ is%
\begin{equation}
b=(2\pi\hbar)^{n}W\phi_{0}\ast a \label{wwsymbol0}%
\end{equation}
that is
\begin{equation}
b(z)=2^{n}\int_{\mathbb{R}^{2n}}e^{-\frac{i}{\hbar}|z-z^{\prime}|^{2}%
}a(z^{\prime})dz^{\prime}. \label{wwsymbol}%
\end{equation}

\end{proposition}

\begin{proof}
Let $\pi(z_{0})$ be the Weyl symbol of the orthogonal projection (\ref{piz});
we thus have,in view of definition (\ref{defA}),
\[
(\widehat{\Pi}_{0}(z)\psi|\phi)_{L^{2}}=\int_{\mathbb{R}^{2n}}\pi
(z_{0})(z^{\prime})W(\psi,\phi)(z^{\prime})dz^{\prime}%
\]
for all $\psi,\phi\in\mathcal{S}(\mathbb{R}^{n})$ (see \textit{e.g}.
\cite{Birkbis}, \S 10.1.2); the translational covariance of the Wigner
transform (\cite{Birkbis}, \S 9.2.2) that is, we have,
\[
\pi(z_{0})(z^{\prime})=(2\pi\hbar)^{n}W(\widehat{T}(z)\phi_{0})(z^{\prime
})=(2\pi\hbar)^{n}W\phi_{0}(z-z^{\prime})
\]
and hence%
\[
(\widehat{\Pi}_{0}(z)\psi|\phi)_{L^{2}}=(2\pi\hbar)^{n}\int_{\mathbb{R}^{2n}%
}W\phi_{0}(z-z^{\prime})W(\psi,\phi)(z^{\prime})dz^{\prime}~.
\]
Using Fubini's theorem we get for $\psi,\phi\in\mathcal{S}(\mathbb{R}^{n})$%
\begin{align*}
(\widehat{A}\psi|\phi)_{L^{2}}  &  =\int_{\mathbb{R}^{2n}}a(z)(\widehat{\Pi
}_{0}(z)\psi|\phi)_{L^{2}}dz\\
&  =(2\pi\hbar)^{n}\int_{\mathbb{R}^{2n}}a(z)\left[  \int_{\mathbb{R}^{2n}%
}W\phi_{0}(z-z^{\prime})W(\psi,\phi)(z^{\prime})dz^{\prime}\right]  dz\\
&  =(2\pi\hbar)^{n}\int_{\mathbb{R}^{2n}}\left[  \int_{\mathbb{R}^{2n}%
}a(z)W\phi_{0}(z-z^{\prime})dz\right]  W(\psi,\phi)(z^{\prime})dz^{\prime}%
\end{align*}
hence the Weyl symbol of $\widehat{A}$ is
\[
b(z)=(2\pi\hbar)^{n}\int_{\mathbb{R}^{2n}}a(z)W\phi_{0}(z-z^{\prime
})dz^{\prime}=(2\pi\hbar)^{n}(a\ast W\phi_{0})(z)..
\]
Formula (\ref{wwsymbol}) follows in view of the identity \cite{Bast,Birkbis}%
\begin{equation}
W\phi_{0}(z)=(\pi\hbar)^{-n}e^{-\frac{1}{\hbar}|z|^{2}}~. \label{wifo}%
\end{equation}

\end{proof}

\begin{remark}
Note that it follows from formula (\ref{wwsymbol}) that the Weyl symbol $b$ is
a real analytic function, hence we cannot expect an arbitrary Weyl operator to
be an anti-Wick operator \cite{sh87}.
\end{remark}

\subsection{The Feichtinger algebra and its dual\label{subsec13}}

The Feichtinger algebra and its dual are the simplest examples of modulation
spaces. They were introduced in the early 1980's by H. Feichtinger
\cite{Hans1,Hans2}, and have since played an increasingly important role in
time-frequency analysis and in Gabor theory; for a full textbook treatment see
Gr\"{o}chenig's treatise \cite{Gro}; in \cite{Jakob} Jakobsen gives an
up-to-date review of the Feichtinger algebra. Modulation spaces were
originally defined in terms of the short-time Fourier transform (or Gabor
transform) widely used in time-frequency analysis; we have redefined them
in\ \cite{Birkbis} in terms of the cross-Wigner transform, which is more
flexible, and has the indisputable advantage that the metaplectic invariance
of modulation spaces becomes immediately obvious. We are following here our
treatment in \cite{Birkbis}, Chapter 16 and 17.

By definition the Feichtinger's algebra $M^{1}(\mathbb{R}^{n})$ (sometimes
denoted $S_{0}(\mathbb{R}^{n})$) consists of all distributions $\psi
\in\mathcal{S}^{\prime}(\mathbb{R}^{n})$ such that $W(\psi,\phi)\in
L^{1}(\mathbb{R}^{2n})$ for \textit{some} window $\phi\in\mathcal{S}%
(\mathbb{R}^{n})$; when this is the case we have $W(\psi,\phi)\in
L^{1}(\mathbb{R}^{2n})$ for \textit{all} windows $\phi\in\mathcal{S}%
(\mathbb{R}^{n})$ and the formula
\begin{equation}
||\psi||_{\phi}=\int_{\mathbb{R}^{2n}}|W(\psi,\phi)(z)|dz=||W(\psi
,\phi)||_{L^{1}} \label{normphi}%
\end{equation}
defines a norm on the vector space $M^{1}(\mathbb{R}^{n})$; another choice of
window $\phi^{\prime}$ the leads to an equivalent norm and one shows that
$M^{1}(\mathbb{R}^{n})$ is a Banach space for the topology thus defined. We
have the following continuous inclusions:
\begin{equation}
M^{1}(\mathbb{R}^{n})\subset C^{0}(\mathbb{R}^{n})\cap L^{2}(\mathbb{R}%
^{n})\cap L^{1}(\mathbb{R}^{n})\cap FL^{1}(\mathbb{R}^{n}) \label{inclusions}%
\end{equation}
and $\mathcal{S}(\mathbb{R}^{n})$ is dense in $M^{1}(\mathbb{R}^{n})$.
Moreover, $\psi\in L^{2}(\mathbb{R}^{n})$ is in $M^{1}(\mathbb{R}^{n})$ if and
only if $||W\psi||_{L^{1}}<\infty$.

Let $\widehat{S}\in\operatorname*{Mp}(n)$ (the metaplectic group) cover
$S\in\operatorname*{Sp}(n)$; then $W(\widehat{S}\psi,\widehat{S}\phi
)=W(\psi,\phi)\circ S^{-1}$; it follows from this covariance formula and the
fact that the choice of window $\phi$ is irrelevant, that $\widehat{S}\psi\in
M^{1}(\mathbb{R}^{n})$ if and only if $S\psi\in M^{1}(\mathbb{R}^{n})$
(metaplectic invariance $M^{1}(\mathbb{R}^{n})$). It follows in particular
that $M^{1}(\mathbb{R}^{n})$ is invariant by Fourier transform, so it follows
from the second inclusion (\ref{inclusions}) that we have
\begin{equation}
M^{1}(\mathbb{R}^{n})\subset FL^{1}(\mathbb{R}^{n})\cap L^{1}(\mathbb{R}^{n}).
\label{111}%
\end{equation}
As a consequence, using the Riemann--Lebesgue lemma, every $\psi\in
M^{1}(\mathbb{R}^{n})$ is bounded and vanishes at infinity. It also follows
from the metaplectic invariance property that $M^{1}(\mathbb{R}^{n})$ is
stable under linear changes of variables: suppose $L\in GL(n,\mathbb{R})$;
then the operator $\widehat{M}_{L,m}$ defined by $\widehat{M}_{L,m}%
\psi(x)=i^{m}\sqrt{\det L}\psi(Lx)$ for a choice of $m$ $\operatorname{mod}4$
corresponding to the argument of $\det L$ is in $\operatorname*{Mp}(n)$; if
$\psi\in M^{1}(\mathbb{R}^{n})$ we have $\widehat{M}_{L,m}\psi\in
M^{1}(\mathbb{R}^{n})$.

It turns out that $M^{1}(\mathbb{R}^{n})$ is in addition an algebra for both
pointwise product and convolution; in fact if $\psi,\psi^{\prime}\in
M^{1}(\mathbb{R}^{n})$ then $||\psi\ast\psi^{\prime}||_{\phi}\leq
||\psi||_{L^{1}}||\psi^{\prime}||_{\phi}$ so we also have
\begin{equation}
M^{1}(\mathbb{R}^{n})\ast M^{1}(\mathbb{R}^{n})\subset M^{1}(\mathbb{R}^{n}).
\label{m1m1l1}%
\end{equation}
Taking Fourier transforms we conclude that $M^{1}(\mathbb{R}^{n})$ is also
closed under pointwise product.

Replacing $\mathbb{R}^{n}$ with $\mathbb{R}^{2n}$ elements of the Feichtinger
algebra can be viewed as pseudodifferential symbols; the following result was
proven by Gr\"{o}chenig in \cite{Gro96} (Theorem 3); also see Gr\"{o}chenig
and Heil \cite{GH04} or Cordero and Gr\"{o}chenig \cite{cogr0}; it is the
announced refinement of Proposition \ref{propduwong}:

\begin{proposition}
\label{Prop4}Let $a\in M^{1}(\mathbb{R}^{2n})$, then $\widehat{A}%
=\operatorname*{Op}\nolimits_{\mathrm{W}}(a)$ is of trace class and we have%
\[
\operatorname*{Tr}\widehat{A}=\left(  \tfrac{1}{2\pi\hbar}\right)  ^{n}%
\int_{\mathbb{R}^{2n}}a(z)dz.
\]

\end{proposition}

The dual Banach space of the Feichtinger algebra $M^{1}(\mathbb{R}^{n})$ is
denoted by $M^{\infty}(\mathbb{R}^{n})$; it is the space of tempered
distributions consisting of all $\psi\in S^{\prime}(\mathbb{R}^{n})$ such that
$W(\psi,\phi)\in L^{\infty}(\mathbb{R}^{2n})$ for one (and hence all) windows
$\phi\in M^{1}(\mathbb{R}^{n})$. The duality bracket is given by the pairing
\begin{equation}
(\psi,\psi^{\prime})=\int_{\mathbb{R}^{2n}}W(\psi,\phi)(z)\overline
{W(\psi^{\prime},\phi)(z)}dz.
\end{equation}
(It follows from the fact that $L^{\infty}(\mathbb{R}^{2n})$ is the dual space
of $L^{1}(\mathbb{R}^{2n})$; see \cite{Gro}, \S 11.3.) The formula
\[
||\psi||_{\phi}^{\prime}=\sup_{z\in\mathbb{R}^{2n}}|W(\psi,\phi)(z)|=||W(\psi
,\phi)||_{L^{\infty}}<\infty
\]
defines a norm on $M^{\infty}(\mathbb{R}^{n})$. Since $M^{1}(\mathbb{R}^{n})$
is the smallest Banach space isometrically invariant under the action of the
metaplectic group its dual $M^{\infty}(\mathbb{R}^{n})$ is essentially the
largest space of distributions with this property.

\section{Toeplitz States}

\subsection{Toeplitz operators and their Weyl symbols}

We defined in Section \ref{subsec12} the anti-Wick operator $\widehat{A}%
_{\mathrm{AW}}=\operatorname*{Op}\nolimits_{\mathrm{AW}}(a)$ by
\begin{equation}
\widehat{A}_{\mathrm{AW}}=\int_{\mathbb{R}^{2n}}a(z)\widehat{\Pi}_{0}(z)dz
\end{equation}
where $\widehat{\Pi}_{0}(z)$ is the orthogonal projection on the ray
$\mathbb{C}(\widehat{T}(z)\phi_{0})$ and $\phi_{0}$ is the standard Gaussian
(\ref{standard}). The notion of Toeplitz operator generalizes this definition
to arbitrary windows $\phi\in L^{1}(\mathbb{R}^{n})$:

\begin{definition}
Let $\phi\in M^{1}(\mathbb{R}^{n})$ and $a\in L^{1}(\mathbb{R}^{2n})$. The
Toeplitz operator $\widehat{A}_{\phi}=\operatorname*{Op}\nolimits_{\phi}(a)$
with window $\phi$ and symbol $a$ is%
\begin{equation}
\widehat{A}_{\phi}=\int_{\mathbb{R}^{2n}}a(z)\widehat{\Pi}_{\phi}(z)dz
\label{Toeplitz1}%
\end{equation}
where $\widehat{\Pi}_{\phi}:L^{2}(\mathbb{R}^{n})\longrightarrow
L^{2}(\mathbb{R}^{n})$ is the orthogonal projection onto the ray
$\mathbb{C}(\widehat{T}(z)\phi)$.
\end{definition}

Explicitly, for $\psi\in\mathcal{S}(\mathbb{R}^{n}).$%
\begin{equation}
\widehat{A}_{\phi}\psi=\int_{\mathbb{R}^{2n}}a(z)(\psi|\widehat{T}(z_{0}%
)\phi)_{L^{2}}\widehat{T}(z_{0})\phi dz_{0}. \label{Toeplitz2}%
\end{equation}

We will discuss the convergence of the integral (\ref{Toeplitz2}) in a moment,
but we first note that in view of the formulas (\ref{crossamb}) and
(\ref{ambwig}) we can\ rewrite the definition (\ref{Toeplitz1}) in the two
equivalent forms
\begin{align}
\widehat{A}_{\phi}\psi &  =(2\pi\hbar)^{n}\int_{\mathbb{R}^{2n}}%
a(z)\operatorname{Amb}(\psi,\phi)(z)\widehat{T}(z)\phi dz\label{MdG1}\\
\widehat{A}_{\phi}\psi &  =(\pi\hbar)^{n}\int_{\mathbb{R}^{2n}}a(z)W(\psi
,\phi^{\vee})(\tfrac{1}{2}z)\widehat{T}(z)\phi dz; \label{MdG2}%
\end{align}
the second relation is essentially the definition of the single-windowed
Toeplitz (or localization) operators given in the time-frequency analysis
literature (see \textit{e.g.} \cite{cogr,coro,cogr0,Toft,togr}).

The following statement is the analogue of Proposition \ref{Prop2} in the
framework of Toeplitz operators:

\begin{proposition}
\label{Prop3}Let $\phi\in M^{1}(\mathbb{R}^{n})$ and $a\in L^{1}%
(\mathbb{R}^{2n})$. The Toeplitz operator
\[
\widehat{A}_{\phi}=\operatorname*{Op}\nolimits_{\phi}(a)=\int_{\mathbb{R}%
^{2n}}a(z_{0})\widehat{\Pi}_{\phi}(z_{0})dz_{0}%
\]
has Weyl symbol
\begin{equation}
a_{\phi}=(2\pi\hbar)^{n}(a\ast W\phi), \label{WeylToep}%
\end{equation}
that is,
\begin{equation}
\widehat{A}_{\phi}=(2\pi\hbar)^{n}\operatorname*{Op}\nolimits_{\mathrm{W}%
}(a\ast W\phi). \label{Toeplitz3}%
\end{equation}

\end{proposition}

\begin{proof}
Let us determine the Weyl symbol $\pi_{\phi}(z_{0})$ of the orthogonal
projection $\widehat{\Pi}_{\phi}(z_{0})$. It is easily seen, using
(\ref{Toeplitz2}), that the kernel of $\widehat{\Pi}_{\phi}(z_{0})$ is the
function
\begin{align*}
K_{\phi}(x,y)  &  =\widehat{T}(z_{0})\phi(x)\overline{\widehat{T}(z_{0}%
)\phi(y)}\\
&  =e^{-\frac{i}{\hbar}p_{0}(y-x)}\phi(x-x_{0})\overline{\phi(y-x_{0})}%
\end{align*}
hence, by formula (\ref{axp}),%
\[
\pi_{\phi}(z_{0})(z)=\int_{\mathbb{R}^{n}}{}^{-\tfrac{i}{\hbar}(p-p_{0})y}%
\phi(x-x_{0}+\tfrac{1}{2}x)\overline{\phi_{0}(y-x_{0}-\tfrac{1}{2}x)}dy
\]
that is
\[
\pi_{\phi}(z_{0})(z)=(2\pi\hbar)^{n}W\phi(z-z_{0}).
\]
It follows, using (\ref{defA}), that
\[
(\widehat{\Pi}_{\phi}(z)\psi|\chi)_{L^{2}}=(2\pi\hbar)^{n}\int_{\mathbb{R}%
^{2n}}W\phi(z-z_{0})W(\psi,\chi)(z_{0})dz_{0}%
\]
for all $\psi,\chi\in\mathcal{S}(\mathbb{R}^{n})$ and hence
\begin{align*}
(\widehat{A}_{\phi}\psi|\chi)_{L^{2}}  &  =\int_{\mathbb{R}^{2n}%
}a(z)(\widehat{\Pi}_{\phi}(z)\psi|\chi)_{L^{2}}dz\\
&  =(2\pi\hbar)^{n}\int_{\mathbb{R}^{2n}}a(z)\left[  \int_{\mathbb{R}^{2n}%
}W\phi(z-z_{0})W(\psi,\chi)(z_{0})dz_{0}\right]  dz.
\end{align*}
Using the Fubini--Tonnelli theorem we get
\begin{align*}
(\widehat{A}_{\phi}\psi|\chi)_{L^{2}}  &  =\int_{\mathbb{R}^{2n}%
}a(z)(\widehat{\Pi}_{\phi}(z)\psi|\chi)_{L^{2}}dz\\
&  =(2\pi\hbar)^{n}\int_{\mathbb{R}^{2n}}a(z)\left[  \int_{\mathbb{R}^{2n}%
}W\phi(z-z_{0})W(\psi,\chi)(z_{0})dz_{0}\right]  dz\\
&  =(2\pi\hbar)^{n}\int_{\mathbb{R}^{2n}}\left[  \int_{\mathbb{R}^{2n}%
}a(z)W\phi(z-z_{0})dz\right]  W(\psi,\chi)(z_{0})dz_{0}%
\end{align*}
hence the Weyl symbol of $\widehat{A}_{\phi}$ is $a_{\phi}=(2\pi\hbar
)^{n}W\phi\ast\mu$, as claimed.
\end{proof}

\subsection{Toeplitz operators as density operators}

The following result characterizes Toeplitz density operators. It is an
extension to the continuous case of the formula (\ref{discrete1}).

\begin{proposition}
\label{PropMU}Let $\mu\in M^{1}(\mathbb{R}^{2n})$ be a probability density:%
\[
\mu\geq0\text{ \ and \ }\int_{\mathbb{R}^{2n}}\mu(z)dz=1.
\]
For every $\phi\in M^{1}(\mathbb{R}^{n})$ such that $||\phi||_{L^{2}}=1$, the
Toeplitz operator%
\begin{equation}
\widehat{\rho}=(2\pi\hbar)^{n}\operatorname*{Op}\nolimits_{\phi}(\mu
)=(2\pi\hbar)^{n}\operatorname*{Op}\nolimits_{\mathrm{W}}(\mu\ast W\phi)
\label{ToepWeyl}%
\end{equation}
is a density operator on $L^{2}(\mathbb{R}^{n})$.
\end{proposition}

\begin{proof}
The operator $\widehat{\rho}$ is positive semidefinite: by definition
(\ref{Toeplitz2}) we have
\[
(\operatorname*{Op}\nolimits_{\phi}(\mu)\psi|\psi)_{L^{2}}=\int_{\mathbb{R}%
^{2n}}\mu(z)|(\psi|\widehat{T}(z)\phi)_{L^{2}}|^{2}dz
\]
hence $(\operatorname*{Op}\nolimits_{\phi}(\mu)\psi|\psi)_{L^{2}}\geq0$
because $\mu\geq0$ being a probability density. Let us prove that
$\widehat{\rho}$ is of trace class. In view of Proposition \ref{Prop4} it is
sufficient to show that the Weyl symbol $a=(2\pi\hbar)^{n}(\mu\ast W\phi)$ is
in $M^{1}(\mathbb{R}^{2n})$. In view of the convolution algebra property
(\ref{m1m1l1}) of $M^{1}(\mathbb{R}^{2n})$ it is sufficient for this to show
that $M^{1}(\mathbb{R}^{n})$ implies that $W\phi\in M^{1}(\mathbb{R}^{2n})$;
this property is in fact a consequence of a more general result (Prop.2.5 in
\cite{cogr0}), but we give here a direct independent proof. Let $\Phi
\in\mathcal{S}(\mathbb{R}^{2n})$; denoting by $W^{2n}$ the cross-Wigner
transform on $\mathbb{R}^{2n}$ we have, in view of formula (\ref{marge}),%
\[
|\int_{\mathbb{R}^{4n}}W^{2n}(W\phi,\Phi)(z,\zeta)dzd\zeta|=|(W\phi
|\Phi)_{L^{2}(\mathbb{R}^{2n})}|<\infty
\]
and hence $W\phi\in M^{1}(\mathbb{R}^{2n})$ as claimed so that $\widehat{\rho
}$ is trace class. In view of the convolution property (\ref{m1m1l1}) of the
Feichtinger algebra we have $\mu\ast W\phi\in M^{1}(\mathbb{R}^{2n})$ as
desired. Let us finally prove that $\operatorname*{Tr}(\widehat{\rho})=1$. We
have $a\in M^{1}(\mathbb{R}^{2n})\subset L^{1}(\mathbb{R}^{2n})$ hence, by
Proposition \ref{PropTrace},
\[
\operatorname*{Tr}(\widehat{\rho})=\left(  \tfrac{1}{2\pi\hbar}\right)
^{n}\int_{\mathbb{R}^{2n}}a(z)dz=a_{\sigma}(0)~.
\]
Denoting by $F_{\sigma}a$ the symplectic Fourier transform $a_{\sigma}$
(\ref{sft}) we have
\[
a_{\sigma}=(2\pi\hbar)^{n}F_{\sigma}(\mu\ast W\phi)=(2\pi\hbar)^{2n}%
(F_{\sigma}\mu)(F_{\sigma}W\phi)
\]
so that
\[
\operatorname*{Tr}(\widehat{\rho})=(2\pi\hbar)^{2n}(F_{\sigma}\mu
)(0)(F_{\sigma}W\phi)(0).
\]
Since $||\phi||_{L^{2}}=1$, we have
\[
F_{\sigma}W\phi(0)=\left(  \tfrac{1}{2\pi\hbar}\right)  ^{n}\int%
_{\mathbb{R}^{2n}}W\phi(z)dz=\left(  \tfrac{1}{2\pi\hbar}\right)  ^{n}%
\]
and hence
\[
\operatorname*{Tr}(\widehat{\rho})=(2\pi\hbar)^{n}F_{\sigma}\mu(0)=\int%
_{\mathbb{R}^{2n}}\mu(z)dz=1.
\]

\end{proof}

\subsection{Example: Gaussian Toeplitz operators}

Let us return to the Gaussian Wigner distribution (\ref{Gaussian}), which we
write here
\begin{equation}
\rho_{\Sigma}(z)=\left(  \tfrac{1}{2\pi}\right)  ^{n}(\det\Sigma
)^{-1/2}e^{-\frac{1}{2}\Sigma^{-1}z\cdot z}; \label{Gaussianbis}%
\end{equation}
the real positive definite $2n\times2n$ matrix $\ \Sigma$ (the
\textquotedblleft covariance matrix\textquotedblright) satisfies the condition%
\begin{equation}
\Sigma+\frac{i\hbar}{2}J\geq0 \label{sigmabis}%
\end{equation}
which ensures the positivity of the corresponding density operator
$\widehat{\rho}_{\Sigma}=(2\pi\hbar)^{n}\operatorname*{Op}%
\nolimits_{\mathrm{W}}(\rho_{\Sigma})$. We are going to see that the
corresponding Weyl operators $\widehat{\rho}_{\Sigma}=(2\pi\hbar
)^{n}\operatorname*{Op}\nolimits_{\mathrm{W}}(\rho_{\Sigma})$ are Toeplitz
operators (in fact generalized anti-Wick operators) if we assume that $\Sigma$
satisfies a certain condition.

Recall that the symplectic eigenvalues $\lambda_{j}^{\sigma}$ \cite{Birkbis}
of $\Sigma$ are the moduli of the eigenvalues of $J\Sigma$ ($J$ the standard
symplectic matrix); since $J\Sigma$ has the same eigenvalues as the
antisymmetric matrix $\Sigma^{1/2}J\Sigma^{1/2}$ these are of the type $\pm
i\lambda_{j}^{\sigma}$ with $\lambda_{j}^{\sigma}>0$. One proves that
\cite{iceberg,Birkbis}:

\begin{lemma}
\label{Lemma2}The condition (\ref{sigmabis}) is equivalent to $\lambda
_{j}^{\sigma}\geq\frac{1}{2}\hbar$ for $1\leq j\leq n$.
\end{lemma}

We call $\Lambda=(\lambda_{1}^{\sigma},...,\lambda_{n}^{\sigma})$ the
symplectic spectrum of $\Sigma$ (the $\lambda_{j}^{\sigma}$ are usually ranked
in decreasing order: $\lambda_{j+1}^{\sigma}\geq\lambda_{j}^{\sigma}$).
Associated to $\Lambda$ is the Williamson diagonalization of $\Sigma$: there
exists $S\in\operatorname*{Sp}(n)$ such that
\begin{equation}
\Sigma=SDS^{T}\text{ \ },\text{ \ }D=%
\begin{pmatrix}
\Lambda & 0\\
0 & \Lambda
\end{pmatrix}
. \label{Williamson}%
\end{equation}
In particular, if all the symplectic eigenvalues are equal to one then
$\Lambda=\frac{1}{2}\hbar I_{n\times n}$, and $\Sigma$ becomes $\Sigma
_{0}=\frac{1}{2}\hbar SS^{T}$ hence, by formula (\ref{wifo})
\[
\rho_{\Sigma_{0}}(Sz)=(\pi\hbar)^{-n}e^{-\frac{1}{\hbar}|z|^{2}}=W\phi_{0}(z)
\]
where $\phi_{0}$ is the standard \ Gaussian (\ref{standard}). Thus, by the
symplectic covariance of the Wigner transform,
\begin{equation}
\rho_{\Sigma_{0}}(z)=W\phi_{0}(S^{-1}z)=W(\widehat{S}\phi_{0})(z)
\label{rhosigo}%
\end{equation}
where $\widehat{S}\in\operatorname*{Mp}(n)$ is anyone of the two metaplectic
operators covering $S\in\operatorname*{Sp}(n)$.

\begin{proposition}
\label{PropToe}The Weyl operator $\widehat{\rho}_{\Sigma}=(2\pi\hbar
)^{n}\operatorname*{Op}\nolimits_{\mathrm{W}}(\rho_{\Sigma})$ is a Toeplitz
density operator if the symplectic eigenvalues $\lambda_{j}^{\sigma}$ of
$\Sigma$ are all larger than $\frac{1}{2}\hbar$.
\end{proposition}

\begin{proof}
Notice that the conditions $\lambda_{j}^{\sigma}>\frac{1}{2}\hbar$ ensure us
that the condition (\ref{sigmabis}) is satisfied (Lemma \ref{Lemma2}). We know
that $\widehat{\rho_{\Sigma}}$ is a density operator so there remains to show
that $\widehat{\rho}_{\Sigma}$ is Toeplitz, \textit{i.e.} that there exist
$\mu$ and $\phi$ such that $\rho_{\Sigma}=\mu\ast W\phi$. We begin by
remarking that a well-known formula from elementary probability theory says
that if $\Sigma^{\prime}$ and $\Sigma^{\prime\prime}$ are two symmetric real
positive definite $2n\times2n$ matrices then $\rho_{\Sigma^{\prime}%
+\Sigma^{\prime\prime}}=\rho_{\Sigma^{\prime}}\ast\rho_{\Sigma^{\prime\prime}%
}$. Choose in particular for $\Sigma^{\prime}$ the matrix $\Sigma_{0}=\frac
{1}{2}\hbar SS^{T}$ defined above, and $\Sigma^{\prime\prime}=\Sigma
-\Sigma_{0}$; then $\Sigma^{\prime}+$ $\Sigma^{\prime\prime}=\Sigma$. In view
of the diagonalization result (\ref{Williamson}) and the assumption
$\lambda_{j}^{\sigma}>\frac{1}{2}\hbar$ for $1\leq j\leq n$ we have
\begin{equation}
\Sigma^{\prime\prime}=S\left(  D-\tfrac{1}{2}\hbar I_{2n\times2n}\right)
S^{T}>0 \label{sigmaprime2}%
\end{equation}
hence $\rho_{\Sigma^{\prime\prime}}\in\mathcal{S}(\mathbb{R}^{2n})$. On the
other hand $\Sigma^{\prime}=\Sigma_{0}$ implies that $\rho_{\Sigma^{\prime}%
}=W(\widehat{S}\phi_{0})$ in view of formula (\ref{rhosigo}). The proposition
follows taking $\phi=\widehat{S}\phi_{0}$ and $\mu=\rho_{\Sigma^{\prime\prime
}}$.
\end{proof}

\section{Separability Properties of Toeplitz Operators}

We will now use the following notation. We introduce the splitting
$\mathbb{R}^{2n}=\mathbb{R}^{2n_{A}}\oplus\mathbb{R}^{2n_{B}}$ and write:
$z=(z_{A},z_{B})=z_{A}\oplus z_{B}$ with $z_{A}=(x_{1},p_{1},...,x_{n_{A}%
},p_{n_{A}})$ and $z_{B}=(x_{n_{A}+1},p_{n_{A}+1},$ $...,x_{n},p_{n})$. We
equip the symplectic spaces $\mathbb{R}^{2n_{A}}$ and $\mathbb{R}^{2n_{B}}$
with their canonical bases. The symplectic structure on $\mathbb{R}^{2n}$ is
then $\sigma(z,z^{\prime})=Jz\cdot z^{\prime}$ with $J=J_{A}\oplus J_{B}$
where
\[
J_{A}=\bigoplus_{k=1}^{n_{A}}J_{k}\text{ \ },\text{ \ }J_{k}=%
\begin{pmatrix}
0 & 1\\
-1 & 0
\end{pmatrix}
\]
and likewise for $J_{B}$. Thus $J_{A}$ (\textit{resp.} $J_{B}$) determines the
symplectic structure on the partial phase space $\mathbb{R}^{2n_{A}}$
(\textit{resp}. $\mathbb{R}^{2n_{B}}$). We denote by $I_{A}$ the identity
$(x_{A},p_{A})\longmapsto(x_{A},p_{A})$ and by $\overline{I}_{B}$ the
involution $(x_{B},p_{B})\longmapsto(x_{B},-p_{B})$ (\textquotedblleft partial
reflection\textquotedblright).

\subsection{The notion of separability}

A density operator $\widehat{\rho}$ on $L^{2}(\mathbb{R}^{n})$ is
$AB$-\textit{separable} if there exist sequences of density operators
$\widehat{\rho}_{j}^{A}$ on $L^{2}(\mathbb{R}^{n_{A}}))$ and $\widehat{\rho
}_{j}^{B}$ on $L^{2}(\mathbb{R}^{n_{B}})$ ($n_{A}+n_{B}=n$) and real numbers
$\alpha_{j}\geq0$, $\sum_{j}\alpha_{j}=1$ such that
\begin{equation}
\widehat{\rho}=\sum_{j\in\mathcal{I}}\alpha_{j}\widehat{\rho}_{j}^{A}%
\otimes\widehat{\rho}_{j}^{B} \label{AB}%
\end{equation}
where the convergence is for the trace class norm. When $\widehat{\rho}$ is
not separable, it represents a \textit{entangled state} in quantum mechanics
\cite{QHA}. Here is a well-known necessary condition for a density operator to
be $AB$-\textit{separable}. We recall that the transpose of the Weyl operator
$\widehat{A}=\operatorname*{Op}\nolimits_{\mathrm{W}}(a)$ is $\widehat{A}%
^{T}=\operatorname*{Op}\nolimits_{\mathrm{W}}(a\circ\overline{I})$ where
$\overline{I}$ is the involution $(x,p)\longmapsto(x,-p)$ (formula
(\ref{transp}). Similarly, one defines the partial transpose $\widehat{A}%
^{T_{B}}$ with respect to the $B$ variables by
\begin{equation}
\widehat{A}^{T_{B}}=\operatorname*{Op}\nolimits_{\mathrm{W}}(a\circ
(I_{A}\oplus\overline{I}_{B})). \label{partransp}%
\end{equation}

For notational simplicity we will write $a\circ\overline{I}_{B}$, $\rho
_{j}\circ\overline{I}_{B}$, \textit{etc}. instead of $a\circ(I_{A}%
\oplus\overline{I}_{B})$, $\rho_{j}\circ(I_{A}\oplus\overline{I}_{B})$,
\textit{etc}. The following result can be found in many physics texts, we give
a rigorous proof thereof below:

\begin{proposition}
\label{PropPPT}Let $\widehat{\rho}$ be a density operator on $L^{2}%
(\mathbb{R}^{n})$. Suppose that the $AB$-separability condition (\ref{AB})
holds. Then the partial transpose%
\begin{equation}
\widehat{\rho}^{T_{B}}=\sum_{j}\alpha_{j}\widehat{\rho}_{j}^{A}\otimes
(\widehat{\rho}_{j}^{B})^{T} \label{rhotb2}%
\end{equation}
is also a density operator.
\end{proposition}

\begin{proof}
(We are following the argument in de Gosson \cite{QHA}, \S 16.2.2). In view of
(\ref{partransp}) the transpose $(\widehat{\rho}_{j}^{B})^{T}$ is explicitly
given by
\begin{equation}
(\widehat{\rho}_{j}^{B})^{T}=(2\pi\hbar)^{n_{B}}\operatorname*{Op}%
\nolimits_{\mathrm{W}}(\rho_{j}\circ\overline{I}_{B}). \label{parttrans}%
\end{equation}
Suppose that the separability condition (\ref{AB}) holds; then the Wigner
distribution of $\widehat{\rho}$ is $\rho=\sum_{j}\lambda_{j}\rho_{j}%
^{A}\otimes\rho_{j}^{B}$ with
\[
\rho_{j}^{A}=\sum_{\ell}\alpha_{j,\ell}W^{A}\psi_{j,\ell}^{A}\text{ \ ,
\ }\rho_{j}^{B}=\sum_{m}\beta_{j,m}W^{B}\psi_{j,m}^{B}%
\]
where $(\psi_{j,\ell}^{A},\psi_{j,\ell}^{B})\in L^{2}(\mathbb{R}^{n_{A}%
})\times L^{2}(\mathbb{R}^{n_{B}})$ and $\alpha_{j,\ell},\beta_{j,m}\geq0$ are
such that $\sum_{\ell}\alpha_{j,\ell}=1$ and $\sum_{m}\beta_{j,m}=1$ ($W^{A}$
and $W^{B}$ are the Wigner transforms in the $z_{A}$ and $z_{B}$ variables,
respectively). Thus
\[
\rho=\sum_{j,\ell,m}\gamma_{j,\ell,m}W^{A}\psi_{j,\ell}^{A}\otimes W^{B}%
\psi_{j,m}^{B}%
\]
where $\gamma_{j,\ell,m}=\lambda_{j}\alpha_{j,\ell}\beta_{j,m}\geq0$. We have
\[
\rho(\overline{I}_{B}z)=\sum_{j\in\mathcal{I}}\lambda_{j}\rho_{j}^{A}%
(z_{A})\rho_{j}^{B}(\overline{I}_{B}z_{B})
\]
and thus%
\[
\rho\circ\overline{I}_{B}=\sum_{j,\ell,m}\gamma_{j,\ell,m}W(\psi_{j,\ell}%
^{A}\otimes\overline{\psi}_{j,m}^{B})
\]
hence $\operatorname*{Op}\nolimits_{\mathrm{W}}(\rho\circ\overline{I}_{B})$ is
also a positive semidefinite trace class operator. That we have
$\operatorname*{Tr}(\widehat{\rho}^{T_{B}})=1$ is obvious.
\end{proof}

The result above is sometimes called in physics the \textquotedblleft PPT
criterion\textquotedblright\ for \textquotedblleft positive partial
transpose\textquotedblright. It is known that while the PPT criterion gives a
necessary condition for separability \cite{Peres} it is also sufficient in the
case $n_{A}n_{B}\leq6$ \cite{horo}.

The problem of finding a general sufficient condition for separability of
density operators is still unsolved.

\subsection{Separability: the Toeplitz case}

We apply Proposition \ref{PropPPT} to Toeplitz density operators. Let us begin
by calculating the partial transpose of the density operator%
\[
\widehat{\rho}=(2\pi\hbar)^{n}\operatorname*{Op}\nolimits_{\mathrm{W}}(\mu\ast
W\phi)\text{ \ , }\mu\in M^{1}(\mathbb{R}^{2n})\text{\ , }\phi\in
M^{1}(\mathbb{R}^{n}).
\]
We have, by formula (\ref{partransp}),
\begin{equation}
\widehat{\rho}^{T_{B}}=(2\pi\hbar)^{n}\operatorname*{Op}\nolimits_{\mathrm{W}%
}((\mu\ast W\phi)\circ\overline{I}_{B}).
\end{equation}
A simple calculation shows that
\[
(\mu\ast W\phi)\circ\overline{I}_{B}=(\mu\circ\overline{I}_{B})\ast(W\phi
\circ\overline{I}_{B}).
\]
Obviously $\mu\circ\overline{I}_{B}\in M^{1}(\mathbb{R}^{2n})$ (the
Feichtinger algebra is closed under linear changes of variables). Let us
examine whether $W\phi\circ\overline{I}_{B}$ is the Wigner transform of some
$\phi^{\prime}\in M^{1}(\mathbb{R}^{n})$. It follows from a general (non-)
covariance result (Theorem 1 in Dias \textit{et al.} \cite{digopra}) that we
cannot expect in general $W\phi\circ\overline{I}_{B}$ \ to be a Wigner
transform. There are however two exception. Assume that $\phi=\phi_{A}%
\otimes\phi_{B}$ with $\phi_{A}\in M^{1}(\mathbb{R}^{n_{A}})$ and $\phi_{B}\in
M^{1}(\mathbb{R}^{n_{B}})$. Then
\[
W\phi\circ\overline{I}_{B}=W_{A}\phi_{A}\otimes W_{B}(\phi_{B}\circ
\overline{I}_{B})=W_{A}\phi_{A}\otimes W_{B}\overline{\phi_{B}}%
\]
so that we have in this case
\begin{equation}
\widehat{\rho}^{T_{B}}=(2\pi\hbar)^{n}\operatorname*{Op}\nolimits_{\mathrm{W}%
}(\mu\circ\overline{I}_{B})\ast(W_{A}\phi_{A}\otimes W_{B}\overline{\phi_{B}%
}). \label{rotb}%
\end{equation}
It follows that:

\begin{proposition}
\label{PropSepToe}Assume that $\phi=\phi_{A}\otimes\phi_{B}$ with $\phi_{A}\in
M^{1}(\mathbb{R}^{n_{A}})$ and $\phi_{B}\in M^{1}(\mathbb{R}^{n_{B}})$. The
partial transpose $\widehat{\rho}^{T_{B}}$ of the Toeplitz density operator
$\widehat{\rho}=(2\pi\hbar)^{n}\operatorname*{Op}\nolimits_{\mathrm{W}}%
(\mu\ast W\phi)$ is also a Toeplitz density operator, given by formula
(\ref{rotb}), that is
\begin{equation}
\widehat{\rho}^{T_{B}}=(2\pi\hbar)^{n}\operatorname*{Op}\nolimits_{\mathrm{W}%
}(\mu\circ\overline{I}_{B})\ast W(\phi_{A}\otimes\overline{\phi_{B}})
\label{rotbb}%
\end{equation}
and we have $\mu\circ\overline{I}_{B}\in M^{1}(\mathbb{R}^{2n})$\ and
$\phi_{A}\otimes\overline{\phi_{B}}\in M^{1}(\mathbb{R}^{n}).$
\end{proposition}

The second case where $W\phi\circ\overline{I}_{B}$ is a Wigner transform is
when the window $\phi$ is a generalized Gaussian
\begin{equation}
\phi_{X,Y}(x)=\left(  \tfrac{1}{\pi\hbar}\right)  ^{n/4}(\det X)^{1/4}%
e^{-\tfrac{1}{2\hbar}(X+iY)x\cdot x}\label{psixy}%
\end{equation}
where $X$ and $A$ are symmetric and $X$ positive definite. The Wigner
transform of this function is well-known \cite{Bast,Birk,Littlejohn} and given
by%
\begin{equation}
W\phi_{X,Y}(z)=(\pi\hbar)^{-n}e^{-\frac{1}{\hbar}Gz\cdot z}\label{wpsixy}%
\end{equation}
where $G$ is the symmetric positive definite symplectic matrix%
\begin{equation}
G=%
\begin{pmatrix}
X+YX^{-1}Y & YX^{-1}\\
X^{-1}Y & X^{-1}%
\end{pmatrix}
\label{G}%
\end{equation}
in the $z=(x,p)$ ordering. One verifies that $G=SS^{T}$ where%
\begin{equation}
S=%
\begin{pmatrix}
X^{1/2} & YX^{-1/2}\\
0 & X^{-1/2}%
\end{pmatrix}
\in\operatorname*{Sp}(n).\label{bi}%
\end{equation}
Setting $\Sigma^{-1}=\frac{2}{\hbar}G$ the function $W\phi_{X,Y}$ is the
Wigner distribution of the Gaussian $\widehat{\rho_{\Sigma}}$ which is a pure
state. Now,
\[
W\phi_{X,Y}(\overline{I}_{B}z)=(\pi\hbar)^{-n}e^{-\frac{1}{\hbar}\overline
{I}_{B}G\overline{I}_{B}z\cdot z}%
\]
and we have
\[
\overline{I}_{B}G\overline{I}_{B}\equiv(I_{A}\oplus\overline{I}_{B}%
)G(I_{A}\oplus\overline{I}_{B})\in\operatorname*{Sp}(n).
\]
After a few calculations and a convenient reordering of the coordinates we
arrive at the fact that $\overline{I}_{B}G\overline{I}_{B}$ is obtained from
$G$ by replacing the matrix $Y$ with the matrix $Y^{\prime}$ defined by
\[
Y^{\prime}(x_{A},x_{B})\cdot(x_{A},x_{B})=Y(x_{A},-x_{B})\cdot(x_{A},-x_{B}).
\]
More intuitively, this amounts to saving that $W\phi_{X,Y}\circ\overline
{I}_{B}=W\phi_{X,Y^{\prime}}$ where $\phi_{X,Y^{\prime}}$ is obtained from
$\phi_{X,Y^{\prime}}$ by taking the partial complex conjugate with respect to
the variables $x_{B}=(x_{1},...,x_{n_{B}}).$ 

\subsection{Application: separability of Gaussian density operators}

Werner and Wolf have proven (\cite{ww}, Prop.1) the following necessary and
sufficient condition for separability: a Gaussian density operator
$\widehat{\rho_{\Sigma}}$ is separable if and only if there exist two partial
covariance matrices $\Sigma_{A}$ and $\Sigma_{B}$ of dimensions $2n_{A}%
\times2n_{A}$ and $2n_{B}\times2n_{B}$ satisfying the conditions
\begin{equation}
\Sigma_{A}+\frac{i\hbar}{2}J_{A}\geq0\text{ \ \textit{and} \ }\Sigma_{B}%
+\frac{i\hbar}{2}J_{B}\geq0 \label{quantAB}%
\end{equation}
and such that%
\begin{equation}
\Sigma\geq\Sigma_{A}\oplus\Sigma_{B}~. \label{ww2}%
\end{equation}
In \cite{digopra2} we have proven with Dias and Prata that this condition is
equivalent to the existence of two symplectic matrices $S_{A}\in
\operatorname*{Sp}(n_{A})$ and and $S_{B}\in\operatorname*{Sp}(n_{B})$ such
that
\begin{equation}
\Sigma\geq\frac{\hbar}{2}(S_{A}S_{A}^{T}\oplus S_{B}S_{B}^{T}). \label{papb1}%
\end{equation}

In \cite{dis1,dis2} we have shown that every Gaussian density operator can be
\textquotedblleft disentangled\textquotedblright\ by a symplectic rotations.
More precisely, we showed that for every $\widehat{\rho_{\Sigma}}=(2\pi
\hbar)^{n}\mathrm{Op}^{\mathrm{W}}(\rho_{\Sigma})$ there exists $U\in U(n)$
such that $\widehat{\rho}_{U}=(2\pi\hbar)^{n}\mathrm{Op}^{\mathrm{W}}%
(\rho_{\Sigma}\circ U)$ is separable. It turns out that the use of the
Toeplitz formalism considerable simplifies the proof:

\begin{proposition}
Let $\widehat{\rho_{\Sigma}}$ be the density operator with Weyl symbol
$(2\pi\hbar)^{n}\rho_{\Sigma}$ where%
\begin{equation}
\rho_{\Sigma}(z)=\left(  \tfrac{1}{2\pi}\right)  ^{n}(\det\Sigma
)^{-1/2}e^{-\frac{1}{2}\Sigma^{-1}z\cdot z}.
\end{equation}
There exists a symplectic rotation $U\in U(n)$ such that
\[
\widehat{\rho_{\Sigma}^{U}}=(2\pi\hbar)^{n}\mathrm{Op}^{\mathrm{W}}%
(\rho_{\Sigma}\circ U^{-1})
\]
is a separable Toeplitz density operator and we have $\widehat{\rho_{\Sigma
}^{U}}=\widehat{U}\widehat{\rho_{\Sigma}}\widehat{U}^{-1}$ where
$\widehat{U}\in\operatorname*{Mp}(n)$ is anyone of the two metaplectic
operators covering $U$.
\end{proposition}

\begin{proof}
Let us write as above $\rho_{\Sigma}=\rho_{\Sigma-\Sigma_{0}-}\ast\rho
_{\Sigma_{0}}$ where $\Sigma_{0}=\tfrac{1}{2}\hbar SS^{T}$, $S\in
\operatorname*{Sp}(n)$ as in the Williamson diagonalization (\ref{Williamson}%
). Since $SS^{T}$ is positive definite and symplectic there exists $U\in U(n)$
such that $SS^{T}=U^{T}\Delta U$ where $\Delta$ is a diagonal matrix whose
diagonal elements are the eigenvalues $\lambda_{1},...,\lambda_{2n}$ of
$SS^{T}$ (\cite{Birk}, Prop. 2.13). We thus have
\[
\Sigma_{0}^{U}=U\Sigma_{0}U^{T}=\tfrac{1}{2}\hbar\Delta
\]
and the positive definite symmetric matrix
\[
\Sigma^{U}=U\Sigma U^{T}=U(\Sigma-\Sigma_{0})U^{T}+\Sigma_{0}^{U}\geq
\Sigma_{0}^{U}%
\]
is the covariance matrix of the Gaussian function $\rho_{\Sigma}^{U}%
=\rho_{\Sigma}\circ U^{-1}$. To the latter corresponds the operator
$\widehat{\rho_{\Sigma}^{U}}$ in view of the symplectic covariance of Weyl
operators \cite{Birk}. Let us prove that $\widehat{\rho_{\Sigma}^{U}}$ is
separable. The eigenvalues of $SS^{T}>0$ are all positive and appear in pairs
$(\lambda,1/\lambda)$. In the $AB$-ordering $\Delta$ has the form
$\Delta=\Delta_{A}\oplus\Delta_{B}$ with
\[
\Delta_{A}=\bigoplus_{k=1}^{n_{A}}\Delta_{k}\text{ \ },\text{ \ }\Delta
_{B}=\bigoplus_{k=n_{A}+1}^{n_{A}+n_{B}}\Delta_{k}%
\]
where $\Delta_{k}=%
\begin{pmatrix}
\lambda_{k} & 0\\
0 & \lambda_{k}^{-1}%
\end{pmatrix}
$ for $k=1,...,n$. Clearly $\Delta_{A}\in\operatorname*{Sp}(n_{A})$ and
$\Delta_{B}\in\operatorname*{Sp}(n_{B})$. Since $\Sigma^{U}\geq\Sigma_{0}^{U}$
it follows that
\[
\Sigma^{U}\geq\tfrac{1}{2}\hbar(\Delta_{A}\oplus\Delta_{B})
\]
hence $\widehat{\rho_{\Sigma}^{U}}$ is separable as claimed in view of
(\ref{papb1}) setting $S_{A}=\Delta_{A}^{1/2}$ and $S_{B}=\Delta_{B}^{1/2}$.
\end{proof}

\begin{remark}
In the physical literature symplectic rotations are called \textquotedblleft
passive linear transformations\textquotedblright\ \cite{wolf}. The result
above can thus be restated by saying that every Gaussian state can be made
separable by a passive linear transformation.
\end{remark}

\begin{acknowledgement}
This work has been financed by the Grant P 33447 N of the Austrian Research
Fund FWF.
\end{acknowledgement}

\begin{acknowledgement}
I would like to extend my gratitude to Hans Feichtinger for several very
useful comments and suggestions about the functional setting of this paper.
Also my greatest thanks to G\'{e}za Giedke, Ludocico Lami, and Fernando
Nicacio for having pointed out errors in an early version of this paper.
\end{acknowledgement}

\end{document}